\documentclass[a4paper,12pt]{article}
\usepackage{test,graphicx,subcaption}
\usepackage[authoryear,round]{natbib}
\usepackage{natbib}
\usepackage{color}
\usepackage[colorlinks,citecolor=blue,urlcolor=magenta]{hyperref}
\usepackage{doi}

\usepackage[inline,shortlabels]{enumitem}
\setlist[enumerate,1]{label=(\roman*)}

\usepackage[margin = 1.25in]{geometry}
\usepackage{setspace}
\title{Ambiguous Cheap Talk}
\author{Longjian Li\thanks{Peking University. Email: \href{mailto:ljliecon@gmail.com}{ljliecon@gmail.com}.} } 

\numberwithin{equation}{section}
\numberwithin{thm}{section}

\begin{document}

\theoremstyle{plain}
\newtheorem{definition}{Definition}
\newtheorem{theorem}{Theorem}
\newtheorem{lemma}{Lemma}
\newtheorem{corollary}{Corollary}
\newtheorem{assumption}{Assumption}
\newtheorem{example}{Example}
\newtheorem{remark}{Remark}
\newtheorem{condition}{Condition}
\newtheorem{proposition}{Proposition}
\newtheorem{conjecture}{Conjecture}
\maketitle

\begin{abstract}
This paper explores how ambiguity affects communication. We consider a cheap talk model in which the receiver evaluates the sender's message with respect to its worst-case expected payoff generated by multiplier preferences. We characterize the receiver's optimal strategy and show that the receiver's posterior action is consistent with his ex-ante action. We find that in some situations, ambiguity improves communication by shifting the receiver's optimal action upwards, and these situations are not rare.
\medskip

\end{abstract}

\section{Introduction}
\par
Decision makers usually need to communicate with better-informed senders to acquire relevant information. In many situations, decision makers face ambiguity interpreted by \cite{strzalecki2011axiomatic}, due to limited information, they can not formulate a single prior,  but they can conjecture a reference probability to be the most possible one and considers many other priors. It is natural for us to ask \textit{will} and \textit{how} ambiguity affects communication and changes the sender's welfare?
\par
To address this question formally, we introduce a cheap talk model with ambiguity: an informed sender sends a message to the ambiguity aversion receiver. The receiver evaluates this message with respect to its worst-case expected payoff generated by multiplier preferences and then takes an action to maximize his worst-case payoff.
\par
Suggested by \cite{hansen2001robust}, multiplier preference is represented as:
\begin{equation}
V_R=\mathop{inf}\limits_{\mathcal{F}} E[v(a,\theta)|m]+\beta R(F||G).    
\end{equation}

Where $R(F||G)$ is called Kullback-Leibler divergence which captures the distance between two probability distributions. $G$ is the receiver's reference probability. Multiplier preference proposes soft constraints on possible probability distributions. It only requires possible probabilities are not too far from the reference probability. It has no restrictions on the possible probabilities' detailed properties (such as the distributions' mean and variance). Compared with hard constraints, such soft constraints enable us to consider more general distribution families and broader strategic issues. 
\par
We first explore a special example: the ambiguous parameter $\beta=0$, and the receiver evaluates the sender's message with maxmin preference. We can interpret this example as the receiver knows nothing about the state of the world and can not even propose a reference probability. We show that consistent with the traditional model, there exist partitional interval equilibria. The receiver can guarantee herself in the worst case by choosing the middle point. The intuition is simple: the worst probability distribution should be the most extreme one, thus a Bernoulli distribution, and the receiver will be indifferent between the two endpoints if she chooses the middle point. The sender is better and more information can be conveyed with the full-ambiguous receiver compared with a bayesian receiver whose prior's conditional mean is smaller than the middle point. 
\par
Next, we consider the receiver's posterior optimal action. %and ex ante optimal action when the ambiguous parameter $\beta>0$. 
We find the worst case posterior probability distribution is generated by dividing the reference probability by a normal distribution, thus the worst case distribution is more dispersed than the reference probability around the receiver's action. Surprisingly, even when multiplier preferences have a KL-divergence term in their utility function, the receiver's optimal action is still the conditional mean of the worst-case probability. We then prove the existence of equilibrium by Brouwer fixed point theorem.
 
\par
As we can not explicitly derive the receiver's action's analytic expression, we want to explore her action's properties. We show receiver's optimal action is asymptotic. To our surprise, her optimal action does not necessarily lie between the middle point (full ambiguity case) and her reference probability's conditional mean (Bayesian case), namely, there exist some special reference probabilities and ambiguous parameters such that the receiver's action is larger than his bayesian optimal action and full-ambiguous optimal action, thus more information can be communicated in this case. We construct a counterexample to show the above statement, using this counterexample, we can show the receiver's optimal action does not have monotone property either.
\par
We make welfare analysis to answer the question mentioned in the beginning: Will ambiguity improve communication? We show that if ambiguity shifts upwards receiver's optimal action and decreases conflict between the sender and the receiver, thus the sender is better off in the ambiguity case than in the bayesian case and more information can be transmitted. We also prove that there exists a bijection from the set of reference probabilities
under which the sender is worse when the receiver faces ambiguity to the set of reference probabilities under which the sender is better
when the receiver faces ambiguity, namely, the two sets have the same cardinality.

\par
We consider the ex-ante ambiguity case as an extension. When the receiver faces ex-ante ambiguity, her optimal ex-ante strategy is to choose her optimal posterior strategy in each interval separately. He considers the worst case probability distribution in the following manner: he first considers the worst-case conditional distributions in each interval which are the same as the posterior ambiguity case, and then he considers the worst interval among them. The receiver is dynamically consistent, her conditional actions when facing ex ante ambiguity are the same as her posterior actions.
\par
The remainder of the paper is organized as follows. In section 2, we review the related literature. In section 3, we introduce the basic model and give the full ambiguity example. In section 4, we solve the posterior ambiguity model and then discuss the receiver's optimal action's properties. In section 5, we solve the ex ante ambiguity model. %In section 5, we analyze a patient-doctor communication application. 
Section 6 concludes.

\section{Related Literature}
\par
This paper joins a recently growing literature studying information transmission with ambiguity. This includes the work of \cite{kellner2017modes}, who simplified communication via a two state and two action model. In a follow up paper, \cite{kellner2018endogenous} introduced endogenous ambiguity in cheap talk model by allowing the ambiguity aversion receiver to perform Ellsbergian randomization. In a similar manner, \cite{beauchene2019ambiguous} studied how the sender can be better off using ambiguous communication devices in a persuasion game with ambiguity-averse players. In our model, ambiguity is exogenous, it is due to the receiver can not formulate a single prior model. While \cite{bose2014mechanism} showed ambiguous communication can implement non-incentive-compatible social choice functions. All these papers show ambiguity can sometimes improve communication and welfare.
\par
According to \cite{epstein2003recursive}, Dynamical inconsistency is always an issue in ambiguity updating, prior by prior updating is available only when the set of prior is rectangular. \cite{beauchene2019ambiguous} did not require rectangularity, and therefore allows for violating dynamical consistency. Our model contributes to this literature by allowing dynamical consistency in communication. It offers a new modeling approach that can be used to make more complex ambiguous communication problems tractable. The decision theory based work about ambiguity updating includes \cite{pacheco2002rule} who axiomatized the full Bayesian updating which is used in our paper,  \cite{gilboa1993updating} who axiomatized maximum likelihood updating as well as \cite{epstein2007learning} who studied ambiguity updating in a more detailed model.
\par
This paper also adds to the broader literature of robustness to distributions (sorted by \cite{carroll2019robustness}). This includes the work of \cite{bergemann2011robust}, \cite{carrasco2018optimal}, \cite{auster2018robust}, \cite{du2018robust}, \cite{brooks2021optimal}, \cite{he2022correlation}.

\section{The Basic Model}
\par
We consider a cheap talk model with ambiguity. There are two players: a sender (he) and a receiver (she). The game begins with the realization of a random state $\theta \in [0,1]$ which is privately observed by the sender. After observing $\theta$, the sender sends a message $m: \Theta \rightarrow M=[0,1]$ to the receiver. The receiver then observes $m$ and takes an action $a \in A$. $a$ can be any real number.
\par
Given state $\theta$ and action $a$, the sender's payoff is $$u(\theta,a)=-(a-\theta-d)^2$$ and the receiver's payoff is $$v(a,\theta)=-(a-\theta)^2.$$ $d$ shows the difference in their preferences, called bias parameter.
\par
Different from the standard model, we assume the receiver is ambiguous about $\theta$'s distribution and evaluates messages with respect to their worst case expected payoff generated by multiplier preferences. We can interpret her ambiguity similar to \cite{strzalecki2011axiomatic}, she is not able to formulate a single-prior model due to limited information. But she can conjecture a reference probability $G$ defined on $[0,1]$ that she thinks to be the most possible distribution, and takes many other priors $F$ into consideration. She thinks the priors that are closer to her reference probabilities are more likely to be the true distribution. Notice that in our model, the receiver's ambiguity is only about the prior distribution, not the sender's strategy.
\par
Considering ambiguity in a dynamic setting is rather challenging, for instance, \cite{epstein2003recursive} pointed out dynamical consistency of maxmin preference can be guaranteed only when the prior set is rectangular. Although \cite{gumen2013dynamically} showed multiplier preferences satisfy dynamical stability\footnote{Namely, if we update the ex ante multiplier preferences in a dynamically consistent manner, then the posterior preferences also belong to multiplier preferences.}, they did not point out specific updating rules. To resolve updating and consistency issues, we consider the receiver's posterior and ex-ante strategy separately and then examine if they are dynamically consistent. We state the receiver's ex ante ambiguity model in the extension part.
\par
Formally, We first consider the posterior ambiguity case, where the receiver behaves as if ambiguity arises after she receives the sender's message $m$. A strategy profile $(m^*(\theta), a^*(m))$ and a belief system constitute an equilibrium if the following conditions hold. First, for $\forall \theta \in [0,1]$, sender's optimal strategy $m^*$ solves
 $$\mathop{max}\limits_{m} u(\theta,a(m))$$
Second, for each $m$, $a^*$ solves
$$\mathop{sup}\limits_{a \in A}\mathop{inf}\limits_{\mathcal{F}} E[v(a,\theta)|m]+\beta R(F_m||G_m)$$
Where $G$ is called the reference probability, it can be interpreted as the receiver's best guess about the true distribution\footnote{The relative entropy $R(F||G)$ (also called Kullback-Leibler divergence), is a measure of how one probability is different from a second. It is a mapping from $\Delta[0,1]$ into $[0,\infty)$ defined by
\begin{equation}
R(F||G)=\left\{
\begin{array}{rcl}
\int_{0}^{1}(log \frac{dF}{dG})dF & & { F \in \Delta^{\sigma}(g)}\\
\infty & & {otherwise}\\
\end{array} \right.
\end{equation} Note that in the receiver's problem, $F_m$ has the same support with $G_m$.} and $G_m$ is the updated reference probability. If $m$ is an equilibrium message and $\Theta^m=\{m(\theta)=m\}$, we have $$g_m(\theta)=\frac{g(\theta)}{\int_{\Theta^m} g(\hat{\theta}) d\hat{\theta}}$$
$$f_m(\theta)=\frac{f(\theta)}{\int_{\Theta^m} f(\hat{\theta}) d\hat{\theta}}$$

Parameter $\beta \in [0,\infty)$ measures the degree of ambiguity she faces towards the distribution, we can also interpret it as the receiver's trust level towards reference probability. Higher values of $\beta$ correspond to more information, less ambiguity about the distribution, and more trust in the reference probability. when $\beta=\infty$, there is no ambiguity about distribution and the receiver believes the reference probability is the true probability distribution of the state of the world. When $\beta=0$, the receiver has no information about the distribution at all, we call this case \textit{Full Ambiguity}.

\par
For simplicity, we characterize the sender's optimal strategy here and then focus on the receiver's problem in the following analysis.
\begin{proposition}
The sender's optimal strategy is partitional. There is a profile of threshold $\{\theta_i\}$ such that $0=\theta_0<\theta_1<...<\theta_{N-1}<\theta_N=1$. If $\theta \in [\theta_{i-1},\theta_i]$, the sender sends message $m_i$ for $i=1, 2, 3,..., N$.
\end{proposition}
\begin{proof}
See appendix.
\end{proof}

\subsection{Full Ambiguity Case}
\par
We first consider the special case when $\beta=0$ in which the receiver has no information about the true distribution of the state of the world and considers all probabilities to be possible. In this case, the receiver's problem reduces to:
$$V_R=\mathop{sup}\limits_{a \in A} \mathop{inf}\limits_{\mathcal{F}} E[v(a,\theta)|m].$$
\par
In the model above, faced with great uncertainty and limited information, can the receiver guarantee herself in the worst distribution case? The answer is yes: she can always choose the middle point to get such a guarantee.
\begin{theorem}
When $\beta=0$, $\theta \in [\theta_{i-1},\theta_i]$ , the receiver's optimal action is to choose the middle point, namely, $a^*=\frac{\theta_{i-1}+\theta_i}{2}$.
\end{theorem}
\begin{proof}
Here we only show the idea of the proof. If the action $a'>a^*=\frac{\theta_i+\theta_{i-1}}{2}$, then the worst case distribution is $p(\theta_{i-1})=1$, namely, the farest endpoint is chosen with probability one. If the action $a''<a^*$, similarly, the worst case distribution is $p(\theta_i)=1$, so the optimal strategy for the receiver is to choose the middle point $a=a^*$ and the worst case probability is $p(\theta_i)=p(\theta_{i-1})=\frac{1}{2}$.  
\end{proof}
\par
It is worth noting that this middle point result not only depends on ambiguity assumption, but also depends on the quadratic utility function form, especially the function's symmetry property. 
\par
The intuition of this result is that when the receiver receives the sender's message, suppose she takes an action other than the middle point, the worst case must be the state of the world is either $\theta_i$ or $\theta_{i-1}$ so that the distance between the action and the state of the world is farthest. To avoid this, the receiver takes the middle point to guarantee his worst case payoff and the nature mixes between $\theta_i$ and $\theta_{i-1}$ with equal probability.
\par
When the conditional mean of true distribution on that interval is smaller than the middle point, the sender will be better and more information can be conveyed compared with a bayesian receiver who knows this true distribution.

\section{Posterior Ambiguity}

\par
Now we consider when the ambiguous parameter $\beta>0$ and the reference probability distribution is $G$. 
In order to use Calculus of Variations, we make following assumption:

\begin{assumption}
$F_m(\theta)$ and $G_m(\theta)$ are continuous and differentiable: $F_m \in C^{2}[0,1]$, $G_m \in C^{2}[0,1]$.
\end{assumption}

\begin{theorem}
When the conditional reference probability is $g_m(\theta)$, the worst case probability distribution is: $$f_m(\theta)=exp(C+\frac{(\theta-a)^2}{\beta})g_m(\theta).$$ Where C is a constant number such that $$\int_{\theta_{i-1}}^{\theta_i} exp(C+\frac{(\theta-a)^2}{\beta})g_m(\theta) d\theta=1 $$
\end{theorem}

\begin{proof}
See appendix.
\end{proof}

\par
Replace $f_m(\theta)$ with the worst case probability formula we show in theorem 2, use some algebra, the receiver's decision problem becomes:
$$\mathop{sup}\limits_{a \in A}  C$$
In the worst case, the receiver tries to maximize the parameter $C$. The receiver's action $a$ can not directly affect his utility, instead, $a$ affects the receiver's decision problem by affecting the worst case probability distribution, and then changing his utility indirectly.
Remember $C$ is a constant number to ensure the generated $f_m(\theta)$ is indeed a probability distribution: $$\int_{\theta_{i-1}}^{\theta_i} exp(C+\frac{(\theta-a)^2}{\beta})g_m(\theta) d\theta=1$$
So Receiver's problem is equivalent to his dual problem
$$\mathop{inf} \limits_{a \in A} \int_{\theta_{i-1}}^{\theta_i} g_m(\theta)exp(\frac{(\theta-a)^2}{\beta}) d\theta$$
The first order condition is 
$$\int_{\theta_{i-1}}^{\theta_i} (\theta-a)exp(\frac{(\theta-a)^2}{\beta})g_m(\theta) d\theta=0 $$
\par
Notice that even though the receiver's optimization problem contains relative entropy, his optimal strategy is still to choose the conditional mean of the generated worst case probability distribution.
\par
Because the above system of equations is non-linear, we can not get an analytical solution directly. However, we can prove the equilibrium always exists by Brouwer fixed point theorem. 
\begin{theorem}
Given $\beta$, there exists a solution $a^*$ of the above system of equations. Thus the equilibrium exists.
\end{theorem}

\begin{proof}
See appendix.
\end{proof}

\par
As we can not explicitly get the receiver's action's analytic expression, it is still instructive to study the solution's properties. We want to explore three questions:
\par 
\begin{itemize}
\item 
Does the solution have asymptotic properties? Namely, when $\beta \rightarrow 0$, will the solution converges to the middle point (\textit{Full Ambiguity} case) and when $\beta \rightarrow \infty$, the solution converges to $g_m(\theta)$'s mean (Bayesian case)?
\item
Does the solution lie between the middle point (\textit{Full Ambiguity} case) and $g_m(\theta)$'s mean (Bayesian case)? 
\item
Does the solution have a monotone property: if $g_m(\theta)$'s mean is larger than the middle point, when $\beta$ increases, will the solution increases?
\end{itemize}
\par
In section 4.1, we will answer the first question. In section 4.2, we will answer the rest questions.

\subsection{Asymptotic Property}
\par
Fortunately, the answer to the first question is yes. When $\beta \rightarrow 0$, the solution converges to the middle point and when $\beta \rightarrow \infty$, the solution converges to $g_m(\theta)$'s mean. We have the following proposition:
\begin{proposition}
When $\beta \rightarrow \infty$, $a \rightarrow \int_{\theta_{i-1}}^{\theta_i}\theta g_m(\theta)d\theta$. When $\beta \rightarrow 0$, $a=\frac{\theta_i+\theta_{i-1}}{2}$.
\end{proposition}

\begin{proof}
See appendix.
\end{proof}
\par
The intuition of this result can be shown like this: When $\beta$ is small and very close to 0, it is very close to the full ambiguity case in which the receiver's optimal strategy is to choose the middle point. When $\beta$ is very large and goes to $\infty$, it is very close to the C-S model case in which the solution is the reference probability's conditional mean.

\subsection{A Counterexample}
\par
It is natural to ask does the solution lie between the middle point (the full ambiguity solution) and $g_m(\theta)$'s mean (the Bayesian solution)? Does the solution have a monotone property: if $g_m(\theta)$'s mean is larger than the middle point, when trust level $\beta$ increases \footnote{ambiguous level decreases}, will the solution increases from the observation of the previous analysis?
\par
However, we have a negative answer to these two questions: we can find a counterexample. For simplicity, denote $g_m^n(\theta)$'s mean as $h^n$. Without loss assume $h^n \geq \frac{\theta_{i-1}+\theta_i}{2}$.
\par
Consider the partitional equilibrium only contains one interval [0,1], $g_m^n(\theta) \in C^2[0,1]$ and $g_m^n(\theta)$  almost uniform
converges to 
\begin{equation}
l(\theta)=
\begin{cases}
0 & \text{$\theta \in [0,\frac{1}{4}]$}\\
3 & \text{$\theta \in [\frac{1}{4},\frac{1}{2}]$ }\\
0 & \text{$\theta \in [\frac{1}{2},\frac{3}{4}]$}\\
1 & \text{$\theta \in [\frac{3}{4},1]$}\\
\end{cases}
\end{equation}
\par
The detailed argument about the existence of $g_m^n(\theta)$ and the proof of the almost uniform convergence see appendix. 
\par
Note that $h^n \rightarrow \frac{1}{2}$. So if the solution lie between the middle point and $g_m^n(\theta)$'s mean, that is $a \in [\frac{\theta_{i-1}+\theta_i}{2},h]$, thus $a \to \frac{1}{2}$. 
\par
Notice $a=\frac{1}{2}$ is $f_m(\theta)$'s mean in interval [0,1]. so we must have 
$$\frac{1}{2} \int_{0}^{1}  f_m(\theta) d\theta=\int_{0}^{1} \theta f_m(\theta) d\theta$$
By computation we have $LHS \neq RHS$ which contradicts with $a$ is $f_m(\theta)$'s conditional mean. So $a \neq \frac{1}{2}$, $a \notin [\frac{\theta_{i-1}+\theta_i}{2},h]$.
\par
From the above argument, we know the solution does not lie between the middle point and $g_m(\theta)$'s conditional mean.
\par
If the solution has a monotone property: $g_m(\theta)$'s conditional mean is not smaller than the middle point, when $\beta$ increases, the solution increases, then $g_m(\theta)$'s conditional mean is her solution's least upper bound and the middle point is her solution's greatest lower bound.
\par
This contradicts our counterexample so the solution does not have a monotone property.
\par
This result is very surprising. It is possible to exist some reference probabilities and ambiguous parameters such that the receiver's action is larger than his bayesian optimal action and full-ambiguous optimal action. Thus more information can be communicated when the receiver is partially ambiguous, partial information is better than full information (Bayesian) and no information (full ambiguity).
\par
Next, we point out some sufficient conditions thus the receiver's optimal strategy is to choose the middle point. When $g_m(\theta)$ is symmetry and its axial of symmetry is the middle point, the receiver's problem solution is still the middle point.
\begin{proposition}
If $g_m(\theta)$ is symmetry and  its axial of symmetry is $\frac{\theta_{i-1}+\theta_i}{2}$, then $a=\frac{\theta_{i-1}+\theta_i}{2}$
\end{proposition}

\begin{proof}
See Appendix.
\end{proof}

%\begin{remark}
%The uniform distribution is a special case of proposition 3. If $g_m(\theta)$ is a normal distribution and its conditional mean $h=\frac{\theta_{i-1}+\theta_i}{2}$, it is also a special case of proposition 3.
%\end{remark}
\begin{remark}
This middle point result not only depends on ambiguity assumption, but also depends on the quadratic utility function form, especially the function's symmetry property. 
\end{remark}

\subsection{Welfare Analysis}
\par
We now turn to the question we asked in the begining: \textit{will} and \textit{how} ambiguity affects communication and changes the sender's welfare? We denote $u^{B}$ as the sender's welfare in the Bayesian model and denote $u^{A}$ as the sender's welfare in the ambiguity model.  $\mathcal{G}_w=\{g: u^B>u^A\}$ is the set of reference probabilities under which the sender is worse when the receiver faces ambiguity, while $\mathcal{G}_b=\{g: u^B<u^A\}$ is the set of reference probabilities under which the sender is better when the receiver faces ambiguity. If ambiguity shifts upwards the best response of the receiver, namely, $a>\int_{\theta_{i-1}}^{\theta_i} \theta g_m(\theta) d \theta$, we have following properties:
\begin{proposition}
 (i) $u^B<u^A$, the sender is better off in the ambiguity case than in the Bayesian case. (ii) Only babbling equilibrium exists when $d>\hat{d}>d^*$ where $d^*$ is the threshold in the traditional model and $\hat{d}$ is the threshold in the ambiguity model. (iii) There exists a bijection from $\mathcal{G}_b$ to $\mathcal{G}_w$, thus $\card (\mathcal{G}_b)=\card(\mathcal{G}_w)$.
\end{proposition}
\begin{proof}
See appendix.
\end{proof}

\par
When ambiguity shifts upwards the receiver's optimal action, this effect will improve communication and encourage information transmission according to Theorem 2 in \cite{chen2015information}. The sender benefits from this effect because he has less conflict of interest with the ambiguous receiver as the receiver's action is closer to his ideal action ($\theta+d$).\footnote{The prior of the sender has no effect on the equilibrium strategies of the game and we allow heterogeneous priors. In order to present the sender's welfare, we assume the sender's prior MLR-dominates the receiver's prior without loss of generality. Sharing the same prior with the receiver is a special of MLR-dominates.} $\hat{d}>d^*$ indicates that more information can be transmitted when the receiver is ambiguity aversion.
\par
The third part of this proposition guarantees that ambiguity can improve the sender's welfare and communication is not rare. It comes from a simple observation: for any reference probability $g(\theta)$ whose induced solution is smaller than its conditional mean, we just construct a probability $\hat{g}(\theta)=g(\theta_{i-1}+\theta_i-\theta)$ that is symmetry with the original probability and the axis of symmetry is the middle point. The constructed reference probability $\hat{g}(\theta)$'s solution is larger than its conditional mean.
\par
We now use an example to show how the receiver's ambiguous attitude $\beta$, her reference probability's mean and variance affect her optimal action, then affect information communication.

\begin{example}
\par
Consider the receiver's reference probability is a normal distribution whose conditional mean is $h$, variance is $\sigma^2$ and ambiguity attitude is $\beta$, denote the middle point of the interval as $m$. If $h<m$, $\beta>2\sigma^2$, then $a>h$, thus ambiguity improves communication. If $h>m$, $\beta>2\sigma^2$, then $a<h$, ambiguity makes sender worse off. If $h<m$, $\beta<2\sigma^2$, then $a>h$. If $h>m$, $\beta<2\sigma^2$, then $a<h$.
\end{example}

\section{Extension: Ex Ante Ambiguity}
\par
Next we consider the receiver's ex-ante strategy, she behaves as if ambiguity arises before she receives the sender's message. A strategy profile $(m^*(\theta), \hat{a}^*(m))$ constitute an equilibrium if the following conditions hold. First, for $\forall \theta \in [0,1]$, sender's optimal strategy $m^*$ solves
 $$\mathop{max}\limits_{m} u(\theta,\hat{a}(m))$$
Second, for each $m$, $\hat{a}^*$ solves
$$\mathop{sup}\limits_{\hat{a} \in A}\mathop{inf}\limits_{\mathcal{F}} E[v(\hat{a},\theta)]+\beta R(\hat{F}||G)$$
\par
Under partitional equilibrium, the receiver's original problem
$$\mathop{sup}\limits_{\hat{a} \in A}\mathop{inf}\limits_{\mathcal{F}} E[v(\hat{a},\theta)]+\beta R(\hat{F}||Q)$$ reduces to
$$\mathop{sup}\limits_{\hat{a_1},..., \hat{a_n}} \mathop{inf}\limits_{(\hat{p_1}, ...,\hat{p_n}), (\hat{F_1},..., \hat{F_n})} E[v(\hat{a},\theta)]+\beta R(\sum_{i=1}^{n} \hat{p_i} \hat{F_i}|| \sum_{i=1}^{n} \hat{p_i} G_i)$$
Where $\hat{a_i}$ is receiver's action in the i-th interval, $\hat{p_i}$ is the probability that she receives the i-th signal, $\hat{F_i}$ is the conditional probability on the i-th interval. We denote the i-th interval as $M_i$. Notice the probability distribution is $\hat{F}=\sum_{i=1}^{n} \hat{p_i} \hat{F_i}$.
\begin{lemma}
$R(\sum_{i=1}^{n} \hat{p}_i \hat{F}_i|| \sum_{i=1}^{n} \hat{p}_i G_i)= \sum_{i=1}^{n} \hat{p}_i R(\hat{F}_i|| G_i)$
\end{lemma}

\begin{proof}
See appendix.
\end{proof}
\par
Because the receiver's problem can be considered seperately in each interval, using the posterior ambiguity's result, the receiver's problem is equivalent to 
$$\mathop{sup}\limits_{\Pi C_i \in \times [\hat{C}_i,C^*_i]} \mathop{inf}\limits_{p_1, ...,p_n} \sum_{i=1}^{n} \hat{p}_i C_i$$
where $\mathop{sup}\limits_{a_i} C_i=C^*_i$, $\mathop{inf}\limits_{a_i} C_i=\hat{C}_i$ , 

\begin{theorem}
The receiver's ex-ante optimal strategy $\hat{a_i}^*$ equals his optimal posterior strategy $a_i^*$ for $i= 1,2,3,...$.The worst case distribution $\hat{p}_i=1$ for $i={argmin}_{i} C^*_i$ and $\hat{p}_i=0$ otherwise. And the conditional probability $\hat{F}_i$ equals his posterior worst case conditional distribution $F_i$.  
\end{theorem}
\par
We note that the receiver's optimal action in the posterior ambiguity case and ex ante ambiguity case is dynamically consistent. The receiver's ex ante optimal strategy is to choose his optimal posterior strategy in each interval separately. This is quite good news for us as dynamical consistency and ambiguity aversion are always thought to have an intrinsic tension between them.
\par
Such consistency makes welfare analysis more easily. As the receiver's ex ante action is the same as her posterior action, we only need to consider her posterior welfare, and her ex ante welfare is the worst interval's posterior welfare.

\section{Conclusion}
\par
This paper explores how ambiguity affects communication. We consider a simple cheap talk model in which the receiver evaluates the informed sender's message with respect to its worst case expected payoff generated by multiplier preferences. 
\par
We show in the full ambiguity example and uniform reference probability example, that the receiver's optimal action is to choose the middle point. The receiver's action is consistent in posterior ambiguity case and ex ante ambiguity case.
\par
We find that ambiguity improves communication by shifting upwards the receiver's optimal action, the number of receivers (the number of reference probabilities) who are ambiguous and more preferred by the sender compared with the bayesian receivers is rich: it is no less than the number of receivers who are ambiguous but less preferred by the sender compared with the bayesian receivers. 
\par
We now discuss promising directions for future research: 
\begin{itemize}
\item 
What are the sufficient conditions for monotone property? Or for all probability distributions, there is no strict monotone property?
\item
What are the necessary conditions for the solution lies between the middle point and $g_m(\theta)$'s conditional mean?
\item
What are the sufficient conditions for the reference probability to improve communication?
\end{itemize}

\section{Appendix}
\subsection{Proof of proposition 1}
\begin{lemma}
 $D^*=\{ d \in D: \sigma^*(m)=d \}$, $D^*$ is finite.
\end{lemma}
\begin{proof}
suppose $a<a'$ are two different actions, by continuity there exists a $\tilde{\theta}$ such that $$u^s(a,\tilde{\theta},d)=u^s(a',\tilde{\theta},d)$$
thus $$a \leq \tilde{\theta}+d \leq a'$$
since all $\theta>\tilde{\theta}$ prefer $a'$ to $a$.
$$a'-a \geq d$$
so $D^*$ is finite.
\end{proof}

\begin{lemma}
$\Theta^*(a)=\{\theta: \sigma^*_R(\sigma^*_S(\theta))=a \}$ is connected.
\end{lemma}
\begin{proof}
 suppose $\theta$, $\theta' \in \Theta^*(a)$ and $\theta < \theta'$. Consider any $\tilde{\theta} \in (\theta, \theta')$, $\tilde{a}$ is the induced action $$-(a-\theta -d)^2 \geq -(\tilde{a}-\theta -d)^2$$ $$-(\tilde{a}-\tilde{\theta} -d)^2 \geq -(a-\tilde{\theta} -d)^2$$ We have $$a \leq \tilde{a}$$ Similarly, $$a \geq \tilde{a}$$ So we have $$a = \tilde{a}$$
\end{proof}

\subsection{Proof of theorem 1}

\begin{lemma}
 if $F(.)$ is non-decreasing, $\theta_{i-1} \leq a \leq \theta_{i}$, there exists a $t$ such that the solution of $\mathop{inf}\limits_{F} \int_{\theta_{i-1}}^{\theta_{i}} (\theta-a)F(\theta)$ satisfies $F(\theta)=t$ $a.e. \ \theta \in [\theta_{i-1},\theta_{i}]$.
\end{lemma}

\begin{proof}
 Because $|(\theta-a)F(\theta)| \leq |\theta_iF(\theta_i)|$, so $$\int_{\theta_{i-1}}^{\theta_i}|(\theta-a)F(\theta)|d\theta \leq \int_{\theta_{i-1}}^{\theta_i}\theta_iF(\theta_i)d\theta< \infty.$$ Similarly, $$\int_{\theta_{i-1}}^{\theta_i}|(\theta-a)F(a)|d\theta \leq \int_{\theta_{i-1}}^{\theta_i}\theta_iF(\theta_i)d\theta < \infty.$$ So $(\theta-a)F(\theta)$ and $(\theta-a)F(a)$ are Lebesgue integrable. Because $F(.)$ is non-decresing, we have $(\theta-a)(F(\theta)-F(a)) \geq 0$. So $$\int_{\theta_{i-1}}^{\theta_i}(\theta-a)(F(\theta)-F(a))d\theta \geq 0$$ If $\int_{\theta_{i-1}}^{\theta_i}(\theta-a)(F(\theta)-F(a))d\theta = 0$, we want to show $F(\theta)=t$ $a.e. \theta \in [\theta_{i-1},\theta_{i}]$. Suppose not, there exists $E=\{ \theta: F(\theta) \neq F(a)\}$ such that $m(E)>0$.
$$\int_{\theta_{i-1}}^{\theta_i}(\theta-a)(F(\theta)-F(a))d\theta \geq \int_{E}(\theta-a)(F(\theta)-F(a))d\theta>0$$ Contradict! So $F(\theta)=t$ $a.e. \theta \in [\theta_{i-1},\theta_{i}]$. 
\end{proof}

\begin{lemma}
$$F(\theta_{i-1})=F(a)$$ 
\end{lemma}

\begin{proof}
Because $$F(\theta)=F(a)$$, for every $\theta \in (\theta_{i-1},\theta_{i})$ and $F(.)$ is right-continuous,
we have $$F(\theta_{i-1})=F(a)$$ 
\end{proof}

\begin{lemma}
The worst case probability distribution is a Bernoulli distribution \footnote{We consider $\theta \in (\theta_{i-1}+\epsilon, \theta_i-\epsilon)$ not $\theta \in (\theta_{i-1}, \theta_i)$ because the Bernoulli distribution is discontinuous and we want to ensure the c.d.f. to be absolute continuous in above interval so that we can integral by parts.}.
\end{lemma}

\begin{proof}
We want to solve $$V=\mathop{sup}\limits_{a \in A}\mathop{inf}\limits_{\mathcal{F}}E[v(a,\theta)|m]$$
Consider $\theta \in (\theta_{i-1}+\epsilon,\theta_i-\epsilon)$\\
$\begin{aligned}
E[\theta|M']
&=\int_{\theta_{i-1}+\epsilon}^{\theta_i-\epsilon} \frac{\theta f(\theta)}{\int_{\theta_{i-1}+\epsilon}^{\theta_i-\epsilon}f(\theta)d\theta}d\theta\\
&=\int_{\theta_{i-1}+\epsilon}^{\theta_i-\epsilon} \frac{\theta dF(\theta)}{F(\theta_i)-F(\theta_{i-1})}\\
&=\frac{(\theta_i-\epsilon)F(\theta_i-\epsilon)-(\theta_{i-1}+\epsilon)F(\theta_{i-1}+\epsilon)}{F(\theta_i-\epsilon)-F(\theta_{i-1}+\epsilon)}-\int_{\theta_{i-1}+\epsilon}^{\theta_i-\epsilon}\frac{F(\theta)}{F(\theta_i-\epsilon)-F(\theta_{i-1}+\epsilon)} d\theta 
\end{aligned}$
\\
\\
$\begin{aligned}
E[{\theta}^2|M']
&=\frac{(\theta_i-\epsilon)^2F(\theta_i-\epsilon)-(\theta_{i-1}+\epsilon)^2F(\theta_{i-1}+\epsilon)}{F(\theta_i-\epsilon)-F(\theta_{i-1}+\epsilon)}-2\int_{\theta_{i-1}+\epsilon}^{\theta_i-\epsilon}\frac{\theta F(\theta)}{F(\theta_i-\epsilon)-F(\theta_{i-1}+\epsilon)} d\theta
\end{aligned}
$
\\
\\
$
\begin{aligned}
\mathop{inf}\limits_{\mathcal{F}} E[v(a,\theta)|m']
&=\mathop{inf}\limits_{\mathcal{F}} -a^2+\frac{2a(\theta_i-\epsilon)F(\theta_i-\epsilon)-2a(\theta_{i-1}+\epsilon)F(\theta_{i-1}+\epsilon)}{F(\theta_i-\epsilon)-F(\theta_{i-1}+\epsilon)}\\
&+\frac{({\theta_{i-1}}+\epsilon)^2F(\theta_{i-1}+\epsilon)-({\theta_i}-\epsilon)^2F(\theta_i-\epsilon)}{F(\theta_i-\epsilon)-F(\theta_{i-1}+\epsilon)}\\
&+2\int_{\theta_{i-1}+\epsilon}^{\theta_i-\epsilon}\frac{\theta F(\theta)-aF(\theta)}{F(\theta_i-\epsilon)-F(\theta_{i-1}+\epsilon)} d\theta 
\end{aligned}
$
\par
We first fix $F(\theta_i-\epsilon)$ and $F(\theta_{i-1}+\epsilon)$ and optimize for $F(\theta)$ where $\theta \in (\theta_{i-1}+\epsilon,\theta_i-\epsilon)$. From lemma 3 and lemma 4 we have $$F(\theta)=F(\theta_{i-1}+\epsilon)$$ 
$$a.e. \quad \theta \in (\theta_{i-1}+\epsilon,\theta_i-\epsilon)$$
\par
Following Carrasco et.al. (2018)'s method, We can think of the receiver's problem as a zero-sum game between the receiver and the Nature in which after observing the sender's messages, the receiver chooses the action, Nature chooses the distribution, Nature's payoff is the negative of the receiver's. The Nash equilibrium of this zero-sum game corresponds to a saddle point of the receiver's problem which is the receiver's maxmin strategy according to the saddle point's properties.
\par
We will show that $a^*=\frac{\theta_{i-1}+\theta_{i}}{2}$
and $p^*=\frac{1}{2}$ is the Nash equilibrium of the zero-sum game by arguing they will not deviate. If the receiver deviates to $a'<a^*$, given Nature's strategy, the receiver is worse off so he will not deviate to $a'<a^*$. Similarly, the receiver will not deviate to $a''>a^*$. If the Nature deviates to $p(\theta_i)>p(\theta_{i-1})$, given the receiver's strategy, Nature's payoff will not change so the Nature will not deviate. The same argument holds for Nature will not deviate to $p(\theta_i)<p(\theta_{i-1})$. 
\end{proof}

\subsection{Proof of theorem 2}
\begin{proof}
$$L(F_m,f_m,\theta)=(2\theta-2a)F_m(\theta)+\beta f_m(\theta)log(f_m(\theta))-\beta f_m(\theta)log(g_m(\theta))$$
$$L_{F_m}(F_m,f_m,\theta)=\frac{\partial L_{f_m}(F_m,f_m,\theta)}{\partial \theta}$$
We have $$2\theta-2a=\frac{{f'}_m(\theta)}{f_m(\theta)}\beta-\frac{{g'}_m(\theta)}{g_m(\theta)}\beta$$
$$\frac{df_m}{f_m}+\frac{2}{\gamma}(a-\theta)d \theta-\frac{dg_m}{g_m}=0$$
$$f_m(\theta)=exp(C+\frac{(\theta-a)^2}{\beta})g_m(\theta).$$
\end{proof}

\subsection{Proof of theorem 3}
\begin{proof}
\par
Define a mapping $L:A  \rightarrow A $ by 
$$L=\frac{\int_{\theta_{i-1}}^{\theta_i} \theta exp(\frac{(\theta-a)^2}{\beta})g_m d\theta}{\int_{\theta_{i-1}}^{\theta_i} exp(\frac{(\theta-a)^2}{\beta})g_m d\theta}$$
This mapping is continuous. The domain of this mapping is $[\theta_{i-1},\theta_i]$ which is a convex compact subset of a Euclidean space. By Brouwer fixed point theorem there exists a fixed point $a^*$ such that $$a^*=\frac{\int_{\theta_{i-1}}^{\theta_i} \theta exp(\frac{(\theta-a)^2}{\beta})g_m(\theta) d\theta}{\int_{\theta_{i-1}}^{\theta_i} exp(\frac{(\theta-a)^2}{\beta})g_m(\theta) d\theta}$$ 
So $a^*$ is a solution to the above equation.
\end{proof}

\subsection{Proof of lemma 1}
\begin{proof}
$$
\begin{aligned}
&R(\sum_{i=1}^{n} p_i F_i|| \sum_{i=1}^{n} p_i G_i)\\
&=\int_{0}^{1} (\sum_{i=1}^{n} p_i f_i) \log(\frac{\sum_{i=1}^{n} p_i f_i}{\sum_{i=1}^{n} p_i g_i}) d\theta\\
&=\sum_{i=1}^{n}\int_{M_i} (\sum_{i=1}^{n} p_i f_i) \log(\frac{\sum_{i=1}^{n} p_i f_i}{\sum_{i=1}^{n} p_i g_i}) d\theta\\
&=\sum_{i=1}^{n}\int_{M_i} p_i f_i \log(\frac{p_i f_i}{p_i g_i}) d\theta\\
&=\sum_{i=1}^{n}p_i\int_{M_i}  f_i \log(\frac{f_i}{g_i}) d\theta\\
&=\sum_{i=1}^{n} p_i R(F_i|| G_i)
\end{aligned}
$$
\end{proof}

\subsection{Proof of proposition 2}
\begin{proof}
When $\beta \to \infty$, $\lim_{\beta \to \infty} exp(C+\frac{(\theta-a)^2}{\beta})=exp(C)=1$, so we have $\lim_{\beta \to \infty} f_m(\theta)=g_m(\theta)exp(C)=g_m(\theta)$.
\par
When $\beta \to 0$,  the worst case probability distribution converges to the Bernoulli distribution in probability, rest arguement is the same as the full ambiguity case.
\end{proof}

\subsection{Existence of Counterexample}
\begin{proof}
We only to consider the discontinuous point of $l(\theta)$.We construct the probability distribution function near $\theta=\frac{1}{4}$, other points can be constructed similarly. Consider the following 
\begin{equation}
g_m^n(\theta)=
\begin{cases}
0 & \text{$\theta \in [0,\frac{1}{4}-\epsilon]$}\\
\frac{7}{5\epsilon^2}(\theta+\epsilon-\frac{1}{4})^2 & \text{$\theta \in [\frac{1}{4}-\epsilon,\frac{1}{4}]$}\\
\frac{14}{5}-\frac{7}{5\epsilon^2}(\theta-\epsilon-\frac{1}{4})^2 & \text{$\theta \in [\frac{1}{4},\frac{1}{4}+\epsilon]$ }\\
\end{cases}
\end{equation}
Notice the constructed $g_m^n(\theta)$ is continuous and differentiable so $g_m^n(\theta) \in C^2[0,1]$. Let $A=[\frac{1}{4}-\epsilon,\frac{1}{4}+\epsilon]$, $\mu(A) \leq 2\epsilon$ and $$\lim_{n \to \infty} \sup_{\theta \notin A} |g_m^n(\theta)-l(\theta)|=0 $$
\par
So we have $g_m^n(\theta)$ almost uniform converges to $l(\theta)$.

\end{proof}

\subsection{Proof of proposition 3}
\begin{proof}
Notice $$f_m(\theta)=exp(C+\frac{(\theta-a)^2}{\beta})g_m(\theta)$$ When $g_m(\theta)$ is symmetry and  its axial of symmetry is $\frac{\theta_{i-1}+\theta_i}{2}$, when $a=\frac{\theta_{i-1}+\theta_i}{2}$, $f_m(\theta)$ is symmetry and  its axial of symmetry is $\frac{\theta_{i-1}+\theta_i}{2}$, then $a=E(\theta|M)$. Thus the middle point is the receiver problem's solution.

\end{proof}

\subsection{Proof of proposition 4}
\begin{proof}
We prove part (ii):
\par
Consider equilibrium with two intervals: $[0,\theta_1]$, $[\theta_1,1]$ and denote the induced actions as $a_1$ and $a_2$. In equilibrium, type $\theta_1$ should be indifferent between $a_1$ and $a_2$: $$\theta_1+d-a_1=a_2-\theta_1-d$$
As $a_1>h_1$, $a_2>h_2$ and $\hat{d}$ is the threshold when $\theta_1=0$, we have $$2\hat{d}=a_1+a_2-2\theta_1=a_1+a_2>h_1+h_2=2d^*$$
so we have $\hat{d}>d^*$

\end{proof}

%\subsection{Proof of theorem 5}

%$\begin{aligned}
%\mathop{sup}\limits_{a \in A} \mathop{inf}\limits_{\mathcal{F}} E[v(a,\theta)|m]+\beta R(F_m||Q_m)&=2a\int_{\theta_{i-1}}^{\theta_i} \theta f_m(\theta)d\theta-\int_{\theta_{i-1}}^{\theta_i} {\theta}^2 f_m(\theta)d\theta \\
%&-a^2+\beta\int_{\theta_{i-1}}^{\theta_i} f_m(\theta)log(f_m(\theta))d\theta\\
%\end{aligned}$
%\\
%\begin{proof}

%We directly use Euler-Lagrange equation to solve this problem. 
%\par
%$$L(F_m,f_m,\theta)=(2\theta-2a)F_m(\theta)+\beta f_m(\theta)log(f_m(\theta))$$
%$$L_{F_m}(F_m,f_m,\theta)=\frac{\partial L_{f_m}(F_m,f_m,\theta)}{\partial \theta}$$
%We have $$2\theta-2a=\frac{{f'}_m(\theta)}{f_m(\theta)}\beta$$
%$$\frac{df_m}{f_m}+\frac{2}{\beta}(a-\theta)d \theta=0$$
%$$\int\frac{df_m}{f_m}+\int\frac{2}{\beta}(a-\theta)d \theta=C$$
%$$f_m=exp(C+\frac{(\theta-a)^2}{\beta})$$
Now we consider the receiver's decision problem:
%$$\mathop{sup}\limits_{a \in A} E[-(a-\theta)^2]+\beta \int_{\theta_{i-1}}^{\theta_i} f_m(\theta)(C+\frac{(\theta-a)^2}{\beta})d\theta$$
%Use some algebra, we have:
%$$\mathop{sup}\limits_{a \in A} -\int_{\theta_{i-1}}^{\theta_i} f_m(\theta)(a-\theta)^2 d\theta+ \int_{\theta_{i-1}}^{\theta_i} f_m(\theta)(\theta-a)^2d\theta+\beta C $$
%So the receiver's decision can reduce to :
%$$\mathop{sup}\limits_{a \in A}  C$$
%Remember $$\int_{\theta_{i-1}}^{\theta_i} exp(C+\frac{(\theta-a)^2}{\beta}) d\theta=1$$
%So Receiver's problem is equivalent to 
%$$\mathop{inf} \limits_{a \in A} \int_{\theta_{i-1}}^{\theta_i} exp(\frac{(\theta-a)^2}{\beta}) d\theta$$
%The first order condition is 
%$$\int_{\theta_{i-1}}^{\theta_i} (\theta-a)exp(\frac{(\theta-a)^2}{\beta}) d\theta=0 $$
%By symmetry the solution is $a=\frac{\theta_{i-1}+\theta_i}{2}$
%\end{proof}

\subsection{Theorem 4}
\begin{proof}
$$\mathop{sup}\limits_{a_1,..., a_n} \mathop{inf}\limits_{(p_1, ...,p_n), (F_1,..., F_n)} \sum_{i=1}^{n} p_i(E[v(a,\theta)|M_i]+R(F_i|| G_i))$$

$$\mathop{sup}\limits_{a_1,..., a_n} \mathop{inf}\limits_{p_1, ...,p_n} \sum_{i=1}^{n} p_i C_i=\mathop{sup}\limits_{\Pi C_i \in \times [\hat{C}_i,C^*_i]} \mathop{inf}\limits_{p_1, ...,p_n} \sum_{i=1}^{n} p_i C_i$$
By saddle point theorem we have following equalities:
$$\mathop{sup}\limits_{\Pi C_i \in \times [\hat{C}_i,C^*_i]} \mathop{inf}\limits_{p_1, ...,p_n} \sum_{i=1}^{n} p_i C_i=\mathop{inf}\limits_{p_1, ...,p_n}\mathop{sup}\limits_{\Pi C_i \in \times [0,C^*_i]} \sum_{i=1}^{n} p_i C_i= \mathop{inf} C^*_i $$
\end{proof}

\bibliography{localbib}

\end{document}